\documentclass[11pt]{article}

\usepackage{amsmath,amssymb,amsthm,setspace,graphicx,array,tikz,url}
\usepackage[text={6.5in,8.4in}, top=1.3in]{geometry}
\newtheorem{theorem}{Theorem}[section]
\newtheorem{proposition}[theorem]{Proposition}
\newtheorem{lemma}[theorem]{Lemma}
\newtheorem{corollary}[theorem]{Corollary}
\newtheorem{conjecture}[theorem]{Conjecture}

\newtheorem{definition}[theorem]{Definition}

\begin{document}
	
\title{
	Finding and counting small tournaments in large tournaments
}
	
\author{
	Raphael Yuster
	\thanks{Department of Mathematics, University of Haifa, Haifa 3498838, Israel. Email: raphael.yuster@gmail.com\;.}
}
	
\date{}
	
\maketitle
	
\setcounter{page}{1}
	
\begin{abstract}
	We present new algorithms for counting and detecting small tournaments in a given tournament.
	In particular, it is proved that every tournament on four vertices (there are four) can be detected in $O(n^2)$ time and counted in $O(n^\omega)$ time where $\omega < 2.373$ is the matrix multiplication exponent. It is also proved that any tournament on five vertices (there are $12$) can be counted in $O(n^{\omega+1})$ time. As for lower-bounds, we prove that for almost all $k$-vertex tournaments, the complexity of the detection problem is not easier than the complexity of the corresponding well-studied counting problem for {\em undirected cliques} of order $k-O(\log k)$.
	
\vspace*{3mm}
\noindent
{\bf 2012 ACM Subject Classification:} Theory of computation $\rightarrow$ Design and analysis of algorithms $\rightarrow$ Graph algorithms analysis

\end{abstract}

\section{Introduction}

Tournaments, which are orientations of complete graphs, are well-studied objects in combinatorics, algorithmic graph theory, 
and computational social choice; see \cite{BG-2008} for a large body of references in all of these areas.
This paper considers the basic problem of detecting and counting small tournaments in larger tournaments.
Detecting and counting small graphs in larger graphs are major topics of research in
algorithmic graph theory (see Subsection \ref{s:related}) and our goal is to investigate these problems
from the tournament perspective.

We usually denote small (fixed size) tournaments by $T$, and large tournaments (problem instances) by $G$. The unique transitive tournament on $k$ vertices is denoted by $T_k$.
For a vertex $v$ in a tournament $G$, $d_G^+(v)$ denotes its out-degree and $d_G^-(v)$ denotes its in-degree;
the subscript is omitted if clear from context. Our main problems are:

\vspace*{5pt}
\noindent
$T$-{\bf DETECT}. For a fixed tournament $T$, determine if an input tournament contains a copy of $T$ as a subgraph.

\vspace*{5pt}
\noindent
$T$-{\bf COUNT}. For a fixed tournament $T$, count the number of copies of $T$ in an input tournament.

\vspace*{5pt}
For a tournament $T$, let $c(T)$ denote the infimum over all reals $t$ such that $T$-COUNT can be computed in $O(n^t)$ time and let $d(T)$ denote the infimum over all reals $t$ such that $T$-DETECT can be solved in $O(n^t)$ time. Let $c(k)$ (resp. $d(k)$) be the maximum of $c(T)$ (resp. $d(T)$) ranging over all $k$-vertex tournaments. Finally, let $c^*(k)$ and $d^*(k)$ denote the corresponding exponents for the counting and detection problems of $K_k$ (the complete graph on $k$ vertices) in undirected graphs.

Some basic observations regarding these parameters follow. Clearly, we must have $d(k) \le c(k)$ and $d^*(k) \le c^*(k)$.
It is well-known \cite{stearns-1959} that in any tournament, every set of $2^{k-1}$ vertices contains
a copy of $T_k$. We therefore have that $T_k$-DETECT can be solved in constant time and hence
$d(T_k)=0$. It is also a simple exercise that the number of $T_3$ is a given tournament $G$
can be obtained by $\sum_{v \in V(G)}\binom{d^+(v)}{2}$. It immediately follows that
$c(T_3) \le 2$ and, in fact, $c(T_3)=2$ as one must read the entire input to determine
the number of $T_3$. From this, it follows that $c(C_3)=2$ where $C_3$ is the directed triangle,
and also $d(C_3)=2$ since it is easy to see that one must read the entire input to determine if a given tournament $G$ is transitive. As $T_3$ and $C_3$ are the only tournaments on three vertices, it follows that $d(3)=c(3)=2$. On the other hand, it is
highly plausible that $d^*(3)=c^*(3)=\omega$ (see Itai and Rodeh \cite{IR-1978} and Vassilevska Williams and Williams \cite{VW-2018} for additional details) where $\omega$ is the matrix multiplication exponent, known to satisfy $2 \le \omega < 2.373$ \cite{AV-2021}.
Note that if $\omega > 2$ then the hypothesis implies that $d^*(3) > c(3)$, i.e., detecting undirected
triangles in graphs is harder than counting any tournament on three vertices in a given tournament.

For $k \ge 4$, the relationship between the parameters $d(k)$, $c(k)$, $d^*(k)$ and $c^*(k)$ (apart from the trivial $d(k) \le c(k)$ and $d^*(k) \le c^*(k)$) is far less obvious. The results in this paper shed light on these relations and establish nontrivial upper bounds for $d(T)$ and $c(T)$ for some small tournaments.

\subsection{Our results}
The first main result consists of two reductions between these parameters.
\begin{theorem}\label{t:1}
	Let $k \ge 3$. Then:\\
	(1) $d(k) \le d^*(k)$.\\
	(2) $d^*(k-O(\log k)) \le  c(k)$. In fact, as $k$ goes to infinity, almost all tournaments $T$ on $k$ vertices satisfy $c(T) \ge d^*(k-O(\log k))$.
\end{theorem}

Before stating our next results, we recall (one of) the definitions of the matrix multiplication exponents. For positive integers $a,b,c$ let $\omega(a,b,c)$ denote the infimum over all reals $t$ such that
any $O(n^a) \times O(n^b)$ matrix can be multiplied with an $O(n^b) \times O(n^c)$ matrix using
$O(n^t)$ field operations. It is well-known that the function $\omega(a,b,c)$ is symmetric.
If the underlying field is finite and its elements can be represented using $O(\log n)$ bits, then field operations translate to (usual) runtime.
Thus, a standard reduction shows that the product of two {\em boolean} matrices with dimensions as above can
be computed in $O(n^t)$ time. The special case $\omega(1,1,1)$ is called {\em the matrix multiplication exponent} and is simply denoted by $\omega$. As mentioned earlier, it is known that $\omega < 2.373$
and clearly $\omega \ge 2$, as trivially $\omega(a,b,c) \ge \max \{a+b,b+c,a+c\}$.

Extending the result of Itai and Rodeh \cite{IR-1978} who proved that $c^*(3) \le \omega$, Ne\v{s}et\v{r}il and Poljak \cite{NP-1985} and Eisenbrand and Grandoni \cite{EG-2004} proved that
$c^*(k) \le \omega(\lfloor k/3 \rfloor,\lceil(k-1)/3\rceil,\lceil k/3 \rceil)$.
In fact, no better bound is known for any $k$, even for $d^*(k)$, and it is conceivable that $\omega(\lfloor k/3 \rfloor,\lceil(k-1)/3\rceil,\lceil k/3 \rceil)$ equals
both $c^*(k)$ and $d^*(k)$ for all $k$.
As shown in Section \ref{sec:reduction}, it is easy to adapt the arguments of \cite{EG-2004,NP-1985} to the tournament setting and obtain that $c(k) \le \omega(\lfloor k/3 \rfloor,\lceil(k-1)/3\rceil,\lceil k/3 \rceil)$.
While this does not imply that $c(k) \le c^*(k)$, we  conjecture that the latter inequality holds,
i.e., that counting any given $k$-vertex tournament is not harder than counting $K_k$.
\begin{conjecture}\label{conj:1}
	For all $k \ge 3$ it holds that $c(k) \le c^*(k)$.
\end{conjecture}
In fact, we conjecture an even stronger statement; that at some point, the two parameters coincide.
\begin{conjecture}\label{conj:2}
	For all sufficiently large $k$ it holds that $c(k) = c^*(k)$.
\end{conjecture}
Given the discussion above, and given that $c(3)=d(3)=2$, it is of interest to investigate the first nontrivial
cases, starting with $k=4$.
\begin{theorem}\label{t:four}
	$d(4) = 2$ and $c(4) \le \omega$.
\end{theorem}
We note  that if $\omega > 2$ and if the inequality in Theorem \ref{t:four} is an equality, then $d(4) < c(4)$ (finding is easier than counting for four-vertex tournaments). 
Furthermore, assuming $d^*(3)=\omega$, counting each tournament on four vertices is
not more difficult than detecting $K_3$ and, assuming $\omega > 2$,
detecting each tournament on four vertices is easier than detecting $K_3$.
Another consequence of Theorem \ref{t:four} is an $O(n^\omega)$ {\em deterministic} algorithm for deciding
whether a given tournament is {\em quasi-random}; see Chung and Graham \cite{CG-1991} for this important and well-studied notion. It has been proved by Lov\'asz (\cite{lovasz-1993} exercise 10.44(b)) and by Coregliano and Razborov \cite{CR-2017} that tournament quasi-randomness can be determined from the number of $T_4$ it contains.

While it is an easy exercise that $c^*(k) \le 1+c^*(k-1)$ and that
$c(T_k) \le 1+c(T_{k-1})$, it is certainly {\em not obvious} that $c(k) \le 1+c(k-1)$.
\begin{conjecture}\label{conj:3}
	For all $k \ge 3$ it holds that $c(k) \le c(k-1)+1$.
\end{conjecture}
A consequence of the following theorem is that Conjecture \ref{conj:3} holds for all $k \le 5$.
\begin{theorem}\label{t:five}
	$c(5) \le \omega+1$.
\end{theorem}
Note that if $d^*(5)=\omega(1,2,2) \ge 4$ then counting each tournament on five vertices (there are $12$
such tournaments) in a given tournament is faster than detecting a $K_5$ in an undirected graph.

The rest of this paper is organized as detailed. Following a subsection on related research, Section \ref{sec:reduction} contains the proof of the reductions yielding Theorem \ref{t:1}. Section \ref{sec:four} considers the case $k=4$ and the proof
of Theorem \ref{t:four}. Section \ref{sec:five} considers the case $k=5$ and the proof of Theorem \ref{t:five}.

\subsection{Related research}\label{s:related}

Detecting and counting (possibly induced) copies of a specified fixed (possibly directed) graph (a.k.a. pattern) in a given host graph is a major topic of research in algorithmic graph theory.
We mention some of the main results in this area that are most relevant to our work.

Itai and Rodeh \cite{IR-1978} proved several decades ago that counting and detection of $K_3$ can be done in $O(n^\omega)$ time, i.e., $c^*(3) \le \omega$. Their method was generalized by
Ne\v{s}et\v{r}il and Poljak \cite{NP-1985} who proved that when $k$ is a multiple of $3$ then
$c^*(k) \le \omega(k/3,k/3,k/3)$ and Eisenbrand and Grandoni \cite{EG-2004} proved that for arbitrary $k \ge 3$, $d^*(k) \le \omega(\lfloor k/3 \rfloor,\lceil(k-1)/3\rceil,\lceil k/3 \rceil)$. In fact, their method works
for counting the number of induced copies of any graph on $k$ vertices.
Cliques in sparser graphs can be found faster. Alon, the author, and Zwick \cite{AYZ-1997} showed how to detect a $K_3$  in time $O(m^{2\omega/(\omega+1)})$ in graphs with $m$ edges. This was generalized to larger cliques by Kloks, Kratsch, and M\"uller \cite{KKM-2000} and by Eisenbrand and Grandoni \cite{EG-2004}. In particular, the latter paper proves that for $k \ge 6$, a $K_k$ can be detected in $O(m^{\omega(\lfloor k/3 \rfloor,\lceil(k-1)/3\rceil,\lceil k/3 \rceil)/2})$ time. If the input graph is further assumed to be of bounded degeneracy, some small patterns can be counted in linear time, while some cannot. A full characterization of such patterns was recently obtained by Bera et al. \cite{BGLSS-2022}, extending earlier results of Chiba and Nishizeki \cite{CN-1985}, Bressan \cite{bressan-2019} and of Bera, Pashanasangi, and Seshadhri \cite{BPS-2020}.

Graphs other than cliques can, in some cases, be detected, and sometimes counted, even faster. Alon, the author, and Zwick \cite{AYZ-1995} proved that for any fixed $k \ge 3$, detecting if a graph has a cycle of length $k$ in a directed or undirected graph can be done in $\tilde{O}(n^\omega)$ time\footnote{$\tilde{O}(.)$ notation is used to suppress sub-polynomial factors.}. The author and Zwick \cite{YZ-1997} showed that even length cycles in undirected graphs can be detected in $O(n^2)$ time.
Alon, the author, and Zwick \cite{AYZ-1997} showed that the number of copies of a given cycle of length at most $7$ can be computed in $\tilde{O}(n^\omega)$ time and Vassilevska Williams et al. \cite{VWWY-2014} showed that at the same time one can detect any induced undirected pattern on four vertices other than $K_4$ and its 
complement.
Plehn and Voigt \cite{PV-1990} proved that if the given $k$-vertex pattern has treewidth $t$, then it can be detected
in $O(n^{t+1})$ time.
Vassilevska Williams and Williams \cite{VW-2013} showed that the number of
copies of a given $k$-vertex graph with an independent set of size $s$ can be computed in
$O(n^{k-s+3})$ without relying on fast matrix multiplication; note that this is faster than exhaustive search.
Using the framework of graph motif parameters, Curticapean, Dell, and Marx
\cite{CDM-2017} showed how to count an undirected pattern on $r$ edges in $n^{0.174r+o(r)}$ time.
Kowaluk, Lingas, and Lundell \cite{KLL-2013} described a general method for counting and detecting small graphs based on linear equations. In particular, their method shows that the detection problem of a given $k$-vertex graph with an independent set of size $s$ can be solved in $O(n^{\omega(\lceil (k-s)/2 \rceil,1,\lfloor (k-s)/2 \rfloor)})$. A counting version with the same runtime was obtained by Floderus et al. \cite{FKLL-2015}. Their method generalizes ideas appearing in the aforementioned papers \cite{AYZ-1997,KKM-2000} for using linear equations to compute subgraph counts. Our result showing that $c(4) \le \omega$ as well as a {\em part} of the result proving $c(5) \le \omega+1$ are also based on the linear equations method.

Finally, we mention a recent result obtained by Dalirrooyfard, Vuong, and Vassilevska Williams \cite{DVW-2019} concerning several {\em lower-bounds} for pattern detection. In particular, they prove that
if the fixed pattern contains $K_k$ as a subgraph, then the time required to detect it is at least the time required to detect a $K_k$. Likewise, if the chromatic number of the pattern is $k$, then (assuming Hadwiger's conjecture), the time required to detect it is at least the time required to detect a $K_k$.

\section{Reductions}\label{sec:reduction}

In this section we prove Theorem \ref{t:1}. The first part of the theorem is proved by the following lemma.

\begin{lemma}\label{l:1}
	For a tournament $T$ on $k$ vertices it holds that $d(T) \le d^*(k)$.
\end{lemma}
\begin{proof}
	We first present a randomized algorithm for $T$-DETECT and then show how to derandomize it.
	Suppose $G$ is an $n$-vertex tournament. Consider a random partition of $V(G)$ into $k$ parts
	$V_1,\ldots,V_k$ (each vertex chooses its part uniformly and independent of other vertices).
	Assume that $V(T)=[k]$ and construct an undirected graph $G^*$ on vertex set $V(G)$ where for each edge $(u,v) \in E(G)$ such that 
	$u \in V_i$, $v \in V_j$ and $(i,j) \in E(T)$, we have that $uv \in E(G^*)$.
	Observe that each $V_i$ is an independent set in $G^*$ and that $G^*$ contains a $K_k$ only if
	$G$ contains $T$. On the other hand, if $G$ contains $T$ as a copy on vertices $v_1,\ldots,v_k$
	where $(v_i,v_j) \in E(G)$ if and only if $(i,j) \in E(T)$, then with probability at least $1/k^k$,
	it holds that $v_i \in V_i$ for $i=1,\ldots,k$. Hence, with probability at least
	$1/k^k$, $G^*$ contains a copy of $K_k$.
	By running a detection algorithm for $K_k$ on $G^*$ we thereby obtain a randomized algorithm for
	$T$-DETECT having the same complexity as $K_k$-DETECT. The algorithm never falsely detects a $T$ if there isn't one, and detects $T$ with probability at least $1/k^k$ if one exists.
	
	To derandomize the algorithm, we need to exhibit a set of partitions of $V(G)$ such that for any ordered set of $k$ distinct vertices $v_1,\ldots,v_k$, at least one of these partitions has $v_i$ in it's $i$'th part for $i=1,\ldots,k$. Such a set of partitions, of size $O(\log n)$, can easily be deterministically constructed in $O(n \log n)$ time, using a sequence of $(k\log k)n$ random variables that are almost
	$(k\log k)$-wise independent \cite{AGHP-1992,NN-1993}; see also \cite{AYZ-1995} for an identical derandomization scenario as the one in this lemma.
	We therefore obtain a deterministic algorithm for $T$-DETECT running in $n^{d^*(k)+o_n(1)}$ time.
	In particular, $d(T) \le d^*(k)$.
\end{proof}

The second part of Theorem \ref{t:1} requires several lemmas and the following definition.
\begin{definition}[signature]
	We say that a subset of vertices $R$ of a tournament $T$ is a {\em signature} if the following holds:
	Suppose $T^*$ is a tournament on the same vertex set as $T$ that is obtained from $T$ by changing
	the orientation of at least one edge with both endpoints not in $R$ but not changing the orientation of any edge with an endpoint in $R$, then $T^*$ is not isomorphic to $T$.
	Let ${\rm sig}(T)$ be the smallest size of a signature of $T$ (vacuously, the entire vertex set of $T$ is a signature). 
\end{definition}

Notice that both ${\rm sig}(T_3)={\rm sig}(C_3)=1$.
In the case of $T_3$, take as signature the vertex that is neither a source nor a sink, and in the case of $C_3$, any vertex can be taken as signature.
It is easy to check that ${\rm sig}(T) > 1$ for any tournament with at least four vertices.
It is also not difficult to check that ${\rm sig}(T_k) \le \lfloor k/2 \rfloor$ ---
indeed consider the vertex labeling of $T_k$ with the labels $[k]$ where
$(i,j)\in E(T_k)$ whenever $i < j$ and take as signature the even labeled vertices.
However, more can be said in general. Let ${\mathcal G}_k$ denote the symmetric probability space of $k$-vertex
tournaments on vertex set $[k]$; i.e. the orientation of each possible edge is determined by a fair coin flip, and all
$\binom{k}{2}$ choices are independent.
\begin{lemma}\label{l:prob}
	Let $T \sim {\mathcal G}_k$. Then $\Pr[{\rm sig}(T) = O(\log k)] = 1-o_k(1)$.
\end{lemma}
\begin{proof}
	Let $r= \lceil 12\log_2 k \rceil$. We prove that with high probability $[r]$ is a signature of $T$.
	To this end, we need to establish some properties that hold with high probability.
	
	Let $H$ be the $r$-vertex sub-tournament of $T$ induced by $[r]$. Let $f$ denote a bijection between $[r]$ and a subset of $r$ vertices
	of $T$ and notice that there are fewer than $k^r$ possible choices for $f$. We say that $f$ is an isomorphism of $H$ whenever it holds that $(i,j) \in E(H)$ if an only if $(f(i),f(j)) \in E(T)$. Let (P1) denote the event that
	no $f$ other than the identity is an isomorphism of $H$.
	
	For a vertex $v \in [k] \setminus [r]$, let $\vec{v}$ be the vector in $\{0,1\}^r$ where $\vec{v}_i = 1$ if and only if
	$(i,v) \in E(T)$ (equivalently, $\vec{ v}_i=0$ if and only if $(v,i) \in E(T)$). Let (P2) be the event that
	all $k-r$ vectors (one for each $v \in [k] \setminus [r]$) are distinct.
	
	We will prove that (P1) and (P2) each hold with high probability. Once we establish that, the lemma follows since if (P1) and (P2) both hold, then $[r]$ is a signature. Indeed,
	suppose that we have changed the orientation of some edges with both endpoints in $[k] \setminus [r]$ to obtain some tournament $T^*$ that is isomorphic to $T$. In particular, suppose that the edge $(v,v')$ with both $v,v' \in [k] \setminus [r]$ changed its orientation.
	Let $g$ be an isomorphism from $T$ to $T^*$. Since $g$ restricted to $[r]$ is an isomorphism of $H$, we have by (P1) that we must have $g(i)=i$ for all $i \in [r]$. But given this, we now have by (P2) that we must have
	$g(v)=v$ for all $v \in [k]\setminus [r]$, so $g$ is the identity. As we assume that $g$ is an isomorphism, we
	have that $(g(v),g(v'))=(v,v')$ has not changed its orientation, a contradiction. Thus, $[r]$ is a signature.
	
	It remains to prove that (P1) and (P2) each hold with high probability. Consider first (P2).
	For distinct $v,v' \in [k] \setminus [r]$, the probability that $\vec{v}=\vec{v'}$ is precisely
	$2^{-r} \le 1/k^{12}$. As there are less than $k^2$ such pairs to consider, it follows that
	(P2) holds with probability at least $1-1/k^{10}$.
	
	To prove (P1), we partition the possible $f$ according to their number of stationary points, where
	notice that if $f$ is not the identity, it has at most $r-1$ stationary points. Fix $1 \le p \le r$ and 
	observe that there are fewer than than $\binom{r}{p}k^{p} < r^pk^p$ possible $f$ with precisely $p$ non-stationary points. Considering such an $f$, let $P \subseteq [r]$ be its set of non-stationary points, where
	$|P|=p$. Consider first the case $1 \le p \le 11$.
	Fix some $v \in P$ and let $R^*=[r] \setminus \{v, f(v)\}$. For $f$ to be an isomorphism of $H$ we need that
	that for each $u \in R^*$, it holds that $(v,u) \in E(H)$ if and only if $(f(v),f(u)) \in E(T)$.
	As these involve $|R^*|$ pairs of {\em distinct} edges (i.e. all $2|R^*|$ involved edges are distinct),
	the probability that $f$ is an isomorphism is at most $1/2^{|R^*|} \le 4/2^r \le 4/k^{12}$.
	Using the union bound for all possible such $f$ (with $1 \le p \le 11$), we obtain that with probability at least
	$1- 11r^{11}k^{11} \cdot 4/k^{12} = 1-o_k(1)$, all such $f$ are not an isomorphism.
	Consider next the case of some fixed $p$ with $12 \le p \le r$. Fix some $q=\lfloor p/4 \rfloor$ vertices of $P$,
	say $v_1,\ldots,v_q$, such that $\{v_1,\ldots,v_q\} \cap \{f(v_1,\ldots,f(v_q)\} = \emptyset$.
	Let $R^*=[r] \setminus \{v_1,\ldots,v_q,f(v_1),\ldots,f(v_q)\}$ and observe that $|R^*| \ge r-2q$.
	For $f$ to be an isomorphism of $H$ we need that
	that for each $1 \le i \le q$ and for each $u \in R^*$, it holds that $(v_i,u) \in E(H)$ if and only if $(f(v_i),f(u)) \in E(T)$.
	As these involve $q|R^*|$ pairs of {\em distinct} edges (i.e. all $2q|R^*|$ involved edges are distinct),
	the probability that $f$ is an isomorphism is at most
	$$
	\frac{1}{2^{q|R^*|}} \le \frac{1}{2^{q(r-2q)}} \le \frac{1}{2^{q(r-p/2)}} \le \frac{1}{2^{qr/2}} \le \frac{1}{k^{6q}}=\frac{1}{k^{6\lfloor p/4 \rfloor}}\;.
	$$
	Using the union bound for all possible such $f$ (with precisely $p$ non-stationary points), we obtain that with probability at least
	$$
	1- r^pk^p \cdot \frac{1}{k^{6\lfloor p/4 \rfloor}} = 1-o(k^{-1})
	$$
	all such $f$ are not an isomorphism. As there are less than $r=o(k)$ cases of $p$ to consider,
	we obtain from the union bound that with probability $1-o_k(1)$, no $f$ other than the identity is an isomorphism of $H$.
\end{proof}

\begin{lemma}\label{l:partition}
	Let $T$ be a fixed tournament on $k$ vertices. Let $G$ be a tournament on $n$ vertices given as input together with a partition of its vertex set into $k$ parts. Then, the number of copies of $T$ in $G$ having precisely
	one vertex in each part can be computed in the same time (up to a constant factor) as that of $T$-COUNT in $G$.
\end{lemma}
\begin{proof}
	Our proof uses inclusion-exclusion and is very similar to the proofs appearing in \cite{CM-2014,GLSY-2022}.
	Suppose the given partition of $V(G)$ is $V_1,\ldots,V_k$. For a non-empty subset $S \subseteq [k]$, let $G_S$ denote the sub-tournament of $G$ on vertex set $\cup_{i \in S}{V_i}$.
	Let $f(G_S)$ denote the number of copies of $T$ in $G_S$ and let $M$ denote the
	number of of copies of $T$ in $G$ with precisely one vertex in each part.
	By the inclusion-exclusion principle we have
	$$
	M = \sum_{S \subseteq [k]} (-1)^{k-|S|}f(G_S)\;.
	$$
	As each $f(G_S)$ is computed by solving $T$-COUNT in a graph with at most $n$ vertices
	and as the last equation involves only a constant number of terms, the lemma follows.
\end{proof}

\begin{lemma}\label{l:reduction}
	Let $T$ be a tournament with $k$ vertices and with ${\rm sig}(T)=r$.
	Then, given an undirected $n$-vertex graph $G$, we can decide if $G$ has a copy of $K_{k-r}$ in time $\tilde{O}(n^{c(T)})$.
\end{lemma}
\begin{proof}
	    Assume that $V(T)=[k]$ where the first $r$ vertices are a signature.
		We present a randomized algorithm for $K_{k-r}$-DETECT. Derandomization is achieved in
		the same way as explained in the proof of Lemma \ref{l:1}, incurring only a logarithmic factor cost.
		Consider a random partition of $V(G)$ into $k-r$ parts
		$V_{r+1},\ldots,V_k$ (each vertex chooses its part uniformly and independent of other vertices).
		Also define new $r$ additional singleton sets $V_1,\ldots,V_r$.
		Construct a {\em tournament} $G^*$ on $n+r=O(n)$ vertices as follows. Its vertex set is
		$\cup_{i=1}^k V_i$. Its edge set is defined as follows. An edge between two vertices in the same part is oriented arbitrarily. Suppose now that $u \in V_i$ and $w \in V_j$ where $i \neq j$ and that
		$(i,j) \in E(T)$. If $i \le r$ or $j \le r$ then orient the edge as $(u,w)$. Otherwise (i.e. if $u > r$ and $v > r$) then: if $uw$ is an edge of $G$ then orient the edge as $(u,w)$ else orient the edge as
		$(w,u)$. This completely defines the tournament $G^*$.
		
		Suppose now that $G^*$ has a copy of $T$ where that copy has precisely one vertex in each of
		$V_1,\ldots,V_k$. By construction of $G^*$ and by the definition of signature, vertex $i$ of $T$ in such a copy must belong to $V_i$ and therefore $G$ has a copy of $K_{k-r}$ with each vertex in precisely one of the sets
		$V_{r+1},\ldots,V_k$. Similarly, if $G$ has no copy of $K_{k-r}$, then it also has no such copy with
		precisely one vertex in each of $V_{r+1},\ldots,V_k$ and so $G^*$ has no copy of $T$ with one vertex in
		each of $V_1,\ldots,V_k$.  By Lemma \ref{l:partition}, we can determine if $G^*$ has a copy of $T$
		with one vertex in each part in time $\tilde{O}(n^{c(T)})$. 
		As in the proof of Lemma \ref{l:1}, observe that if $G$ has a $K_{k-r}$ then with probability at
		least $1/(k-r)^{k-r}$ it holds that each vertex of some such copy belongs to each one of
		$V_{r+1},\ldots,V_k$. We therefore obtain a randomized algorithm for  $K_{k-r}$-DETECT
		running in time  $\tilde{O}(n^{c(T)})$.
\end{proof}

\begin{proof}[Proof of Theorem \ref{t:1}]
	By Lemma \ref{l:1}, $d(T) \le d^*(k)$ for every tournament $T$ on $k$ vertices. Thus,
	$d(k) \le d^*(k)$ and part (1) of the theorem follows.
	By Lemma \ref{l:prob}, almost all tournaments $T$ on $k$ vertices have $sig(T)=O(\log k)$.
	Now, if $T$ is a tournament with $sig(T)=r=O(\log k)$, then by Lemma \ref{l:reduction} it follows that
	$c(T) \ge d^*(k-r)=d^*(k-O(\log k))$. Thus, part (2) of the theorem follows.
\end{proof}

For completion, we end this section by showing that the method of
Ne\v{s}et\v{r}il and Poljak \cite{NP-1985} and Eisenbrand and Grandoni \cite{EG-2004} for counting induced $k$-vertex graphs in undirected graphs
can be used for counting tournaments in a similar way.
\begin{proposition}
	$c(k) \le \omega(\lfloor k/3 \rfloor,\lceil(k-1)/3\rceil,\lceil k/3 \rceil)$.
\end{proposition}
\begin{proof}
	For simplicity, assume that $k$ is a multiple of $3$. The cases of the other two moduli can be proved analogously. Let $T$ be a labeled tournament on vertex set $[k]$.
	Let $A_1=\{1,\ldots,k/3\}$, $A_2=\{k/3+1,\ldots,2k/3\}$ and $A_3=\{2k/3+1,\ldots,k\}$.
	For $X \subseteq \{1,2,3\}$ with $X \neq \emptyset$, let $T_X$ be the labeled sub-tournament of $T$ induced on $\cup_{i \in X}A_i$.
	Let $G$ be a tournament on $n$ vertices whose vertices are labeled with $[n]$.
	For $i=1,2,3$, let $S_i$ be the set of all labeled subgraphs of $G$ that are isomorphic to $T_{\{i\}}$.
	Construct a $0/1$ matrix $Q_1$ whose rows are indexed by $S_1$ and columns indexed by $S_2$.
	We set $Q_1[H,J]=1$ if and only if the labeled sub-tournament of $G$ on $V(H) \cup V(J)$ induces an isomorphic copy of $T_{\{1,2\}}$.
	Similarly, construct a $0/1$ matrix $Q_2$ whose rows are indexed by $S_2$ and columns indexed by $S_3$.
	We set $Q_2[H,J]=1$ if and only if
	the labeled sub-tournament of $G$ on $V(H) \cup V(J)$ induces an isomorphic copy of
	$T_{\{2,3\}}$. Observe that $Q_1$ and $Q_2$ can be constructed in $O(n^{2k/3})$ time by exhaustive search
	and that each have $O(n^{k/3})$ rows and $O(n^{k/3})$ columns.
	Now, consider the product $Q_3=Q_1Q_2$ which, by definition of the matrix multiplication exponent, can be computed in $O(n^{\omega(k/3,k/3,k/3)})$ time. Consider an entry $Q_3[H,J]$ for which $Q_3[H,J] > 0$.
	If it holds that the labeled sub-tournament of $G$ on $V(H) \cup V(J)$ 
	induces an isomorphic copy of $T_{\{1,3\}}$, then
	$Q_3[H,J]$ precisely counts the number of labeled sub-tournaments in $G$ which are isomorphic to
	$T=T_{\{1,2,3\}}$ and in which the vertices of $V(H) \cup V(J)$ correspond to a particular labeled copy of $T_{\{1,3\}}$. Furthermore, each such labeled copy is counted precisely once in this way.
	We can therefore count the number of label-isomorphic copies of $T$ in $G$, and dividing by the order of the automorphism group of $T$ we obtain the solution to $T$-COUNT in $O(n^{\omega(k/3,k/3,k/3)})$ time. As this holds for any given $T$, the proposition follows.
\end{proof}

\section{Finding and counting four-vertex tournaments}\label{sec:four}

There are four distinct tournaments on four vertices. These are $T_4$, $D$, $D^T$ and $X_4$
where $D$ is the tournament with a directed triangle and an additional source vertex,
$D^T$ is its transpose, i.e., the tournament with a directed triangle and an additional sink vertex,
and $X_4$ is the strongly connected tournament on four vertices (i.e., the tournament which has a directed four-cycle). In this section we consider $T$-DETECT and $T$-COUNT for each of them.
As $D^T$ is the transpose of $D$, it immediately follows that $D$-DETECT and $D^T$-DETECT are computationally equivalent and that $D$-COUNT and $D^T$-COUNT are computationally equivalent.
Another trivial observation, mentioned in the introduction, is that $T_k$-DETECT can be solved in constant time for any fixed $k$, in particular $d(T_4)=0$.

Starting with the detection problem, given the above, we need to establish the complexities of $X_4$-DETECT
and $D$-DETECT. The first is fairly simple.
\begin{lemma}\label{l:detect-c4}
	$d(X_4)=2$.
\end{lemma}
\begin{proof}
	We recall an easy graph-theoretic observation: every undirected graph with at least $n+1$ edges has a path of length $3$.
	Suppose we have yet to examine at least $n+1$ edges of the input tournament $G$ and that we have still not found a copy of $X_4$. 
	Then, by the observation above, there are four vertices $a,b,c,d$ such that none of the edges on pairs $ab$, $bc$, $cd$, have been examined. No matter how the other three (possibly examined) edges on the pairs  $ac$, $ad$, $bd$ are oriented, it is immediate to check that it is possible to orient the edges $ab$, $bc$, $cd$ so that $\{a,b,c,d\}$ induce a copy of $X_4$.
	Thus, we might need to examine at least $\binom{n}{2}-n$ edges to determine whether $G$ contain a copy of $X_4$, implying that $d(X_4) \ge 2$.
	
	For the upper bound, compute the strongly connected components of the input tournament in $O(n^2)$ time. If each such component is a singleton or a directed triangle, then the tournament does not have an $X_4$. Otherwise, suppose some strongly connected component has $t \ge 4$ vertices. It is well-known that
	every strongly-connected tournament with $t \ge 3$ vertices is {\em pancyclic} \cite{HM-1966}, i.e., it has a directed cycle of every possible length from $3$ to $t$. In particular, as $t \ge 4$, it has a directed cycle of length $4$. The four vertices of such a cycle induce $X_4$.
\end{proof}

For the proof that $d(D)=2$ we will need a characterization of $D$-free tournaments.
Such a characterization was obtained by Liu in his thesis \cite{liu-2012}: a strongly regular tournament is $D$-free if and only if it is a transitive blowup of the carousel tournament (for every odd $n$, the carousel tournament is the unique regular tournament on $n$ vertices in which the out-neighborhood of each vertex induces a transitive tournament). However, we shall use a different characterization of Gishboliner \cite{gishboliner} whose proof is a bit simpler; it should be noted that both characterizations lead to algorithms implying $d(D)=2$. Suppose a tournament $G$ is not transitive and let $\{a,b,c\}$ induce a directed 
triangle in $G$ with $(a,b),(b,c),(c,a) \in E(G)$. For $S \subseteq \{a,b,c\}$ let
$N_S = \{x \in V(G) \setminus \{a,b,c\} \,:\, \{d \in \{a,b,c\}\,:\, (x,d) \in E(G)\}=S\}$.
So, for example $N_\emptyset$ are all the vertices that are dominated by each of $a,b,c$.
For two disjoint sets of vertices of $G$, we say that a triple of vertices is {\em bad}
in them if the triple induces a $T_3$, where the source and the sink of the $T_3$ are in one of the sets and
the remaining vertex of the $T_3$ is in the other set.
For two disjoint sets of vertices $X,Y$ of $G$, we use $X \rightarrow Y$ to denote that all possible edges go from $X$ to $Y$. Gishboliner's characterization, which we prove for completeness, is the following:
\begin{lemma}\label{l:D-free}
	$G$ is $D$-free if an only if all the following hold:\\
	(1) $N_{\{a,b,c\}}=\emptyset$.\\
	(2) For every $S \subseteq \{a,b,c\}$  with $S \neq \emptyset, \{a,b,c\}$, we have $N_S \rightarrow N_\emptyset$.\\
	(3) For every $S \subseteq \{a,b,c\}$ with $S \neq \{a,b,c\}$, $N_S$ induces a transitive tournament.\\
	(4) $N_{\{a\}} \rightarrow N_{\{b\}} \rightarrow N_{\{c\}} \rightarrow N_{\{a\}}$.\\
	(5) $N_{\{a,b\}} \rightarrow N_{\{b,c\}} \rightarrow N_{\{c,a\}} \rightarrow N_{\{a,b\}}$.\\
	(6) $N_{\{a\}} \rightarrow N_{\{a,b\}} \rightarrow N_{\{b\}}$,
	    $N_{\{b\}} \rightarrow N_{\{b,c\}} \rightarrow N_{\{c\}}$,
	    $N_{\{c\}} \rightarrow N_{\{c,a\}} \rightarrow N_{\{a\}}$.\\
	(7) Each of the pairs $\{N_{\{a\}},N_{\{b,c\}}\}$, $\{N_{\{b\}},N_{\{a,c\}}\}$, $\{N_{\{c\}},N_{\{a,b\}}\}$
	    have no bad triple.
\end{lemma}
\begin{proof}
	We show the necessity of each of the items.\\
	(1) If $x \in N_{\{a,b,c\}}$ then $\{x,a,b,c\}$ induce a copy of $D$.\\
	(2) Assume that there are $x \in N_S$, $y \in N_\emptyset$ and $(y,x) \in E(G)$.
	Since $S$ consists only of one or two vertices of $\{a,b,c\}$, we can take $d,e \in \{a,b,c\}$ such
	that $d \notin S$, $e \in S$ and $(d,e) \in E(G)$. Then $\{d,x,e,y\}$ is a copy of $D$.\\
	(3) Assume that $x,y,z \in N_S$ induce a directed triangle. Take $d \in \{a,b,c\} \setminus S$
	and note that $\{d,x,y,z\}$ induce a copy of $D$.\\
	(4) We prove that $N_{\{a\}} \rightarrow N_{\{b\}}$. The proof of the other cases is symmetrical.
	Assume that there are $x \in N_{\{a\}} $ and $y \in N_{\{b\}}$ such that $(y,x) \in E(G)$.
	Then $\{c,x,a,y\}$ induce a copy of $D$.\\ 
	(5) We prove that $N_{\{a,b\}} \rightarrow N_{\{b,c\}}$. The proof of the other cases is symmetrical.
	Assume that there are $x \in N_{\{a,b\}} $ and $y \in N_{\{b,c\}}$ such that $(y,x) \in E(G)$.
	Then $\{y,b,c,x\}$ induce a copy of $D$.\\ 
	(6) We prove that $N_{\{a\}} \rightarrow N_{\{a,b\}} \rightarrow N_{\{b\}}$. The proof of the other cases is symmetrical. First assume that there are
	$x \in N_{\{a\}}$  and $y \in N_{\{a,b\}}$ such that $(y,x) \in E(G)$.
	Then $\{y,a,b,x\}$ induce a copy of $D$.
	Now assume that there are $x \in N_{\{a,b\}}$  and $y \in N_{\{b\}}$ 
	such that $(y,x) \in E(G)$. Then $\{c,a,y,x\}$ induce a copy of $D$.\\
	(7) We prove for the pair $\{N_{\{a\}},N_{\{b,c\}}\}$. The proof of the other cases is symmetrical.
	Assume first that there are $x_1,x_2 \in N_{\{a\}}$ and $y \in N_{\{b,c\}}$ such that $(x_1,x_2) \in E(G)$, $(x_1,y) \in E(G)$ and $(y,x_2) \in E(G)$.
	Then $\{x_1,x_2,a,y\}$ induce a copy of $D$.
	Analogously, assume that there are $x \in N_{\{a\}}$ and $y_1,y_2 \in N_{\{b,c\}}$ such that
	$(y_1,y_2) \in E(G)$, $(y_1,x) \in E(G)$ and $(x,y_2) \in E(G)$.
	Then $\{y_1,y_2,b,x\}$ induce a copy of $D$.
	
	We next prove that if items 1-7 hold, then $G$ is $D$-free. Let $x \in V(G)$ and let $S(x)$ denote the out-neighbors of $x$. We need to show that $S(x)$ induces a transitive tournament. There are four cases.\\
	(1) $x \in \{a,b,c\}$. By symmetry it is enough to prove for $x = a$.
	We have $S(a)=\{b\} \cup N_{\{b\}} \cup N_{\{c\}} \cup N_{\{b,c\}} \cup N_\emptyset$.
	We have $N_{\{b\}} \rightarrow \{b\}$, $N_{\{b,c\}} \rightarrow \{b\}$, $\{b\} \rightarrow N_{\{c\}}$
	and $\{b\} \rightarrow N_\emptyset$.
	By Item 2 we have $N_{\{b\}} \rightarrow N_\emptyset$, $N_{\{b,c\}} \rightarrow N_\emptyset$,
	$N_{\{c\}} \rightarrow N_\emptyset$.
	By Item 4 we have $N_{\{b\}} \rightarrow N_{\{c\}}$
	and by Item 6 we have $N_{\{b\}} \rightarrow N_{\{b,c\}} \rightarrow N_{\{c\}}$.
	By Item 3 the sets $N_{\{b\}}, N_{\{c\}}, N_{\{b,c\}}, N_\emptyset$ are transitive.
	So the ordering $N_{\{b\}} \rightarrow N_{\{b,c\}} \rightarrow \{b\} \rightarrow N_{\{c\}} \rightarrow 
	N_\emptyset$ is a transitive ordering of $S(a)$.\\
	(2) $x \in N_\emptyset$. Item 2 implies that $S(x) \subset N_\emptyset$ which is transitive by Item 3.\\
	(3) $x \in N_{\{a\}} \cup N_{\{b\}} \cup N_{\{c\}}$. By symmetry it is enough to prove for $x \in N_{\{a\}}$.
	By items 1,2,4,6 we have $S(x) = \{a\} \cup N_{\{b\}} \cup N_{\{a,b\}} \cup N_\emptyset \cup Y \cup Z$
	where $Y = N_{\{a\}} \cap S(x)$ and $Z = N_{\{b,c\}} \cap S(x)$.
	All the sets in this union are transitive.
	By Item 2 we have $N_{\{a,b\}}, N_{\{b\}}, Y, Z \rightarrow N_\emptyset$.
	By Item 6 we have $N_{\{a,b\}} \rightarrow N_{\{b\}}$.
	By Item 7 we have $Y \rightarrow Z$.
	By the definition of $Y,Z$ we have $Y \rightarrow\{a\} \rightarrow Z$.
	By items 4,5,6 we have $Y \rightarrow N_{\{b\}}, N_{\{a,b\}}$ and $N_{\{b\}}, N_{\{a,b\}} \rightarrow Z$. So the ordering $Y \rightarrow N_{\{a,b\}}  \rightarrow \{a\} \rightarrow N_{\{b\}} \rightarrow Z \rightarrow N_\emptyset$ is a transitive ordering of $S(x)$.\\
	(4)  $x \in N_{\{a,b\}} \cup N_{\{b,c\}} \cup N_{\{c,a\}}$. By symmetry it is enough to prove for $x \in N_{\{a,b\}}$. By Items 1,2,5,6 we have $S(x) = \{a,b\} \cup N_{\{b\}} \cup N_{\{b,c\}} \cup N_\emptyset \cup Y \cup Z$ 	where $Y = N_{\{a,b\}} \cap S(x)$ and $Z = N_{\{c\}} \cap S(x)$.
	All the sets in this union are transitive.
	By Item 2 we have $N_{\{b\}}, N_{\{b,c\}}, Y, Z \rightarrow N_\emptyset$.
	By Item 6 we have $N_{\{b\}} \rightarrow N_{\{b,c\}}$.
	By the definition of $Y,Z$ we have $Y \rightarrow\{a,b\} \rightarrow Z$.
	By items 4,5,6 we have $Y \rightarrow N_{\{b\}}, N_{\{b,c\}}$ and $N_{\{b\}}, N_{\{b,c\}} \rightarrow Z$.
	By Item 7 we have $Y \rightarrow Z$.
	So the ordering $Y \rightarrow \{a\} \rightarrow N_{\{b\}} \rightarrow N_{\{b,c\}} \rightarrow \{b\} \rightarrow Z \rightarrow N_\emptyset$ is a transitive ordering of $S(x)$.\\
\end{proof}

\begin{lemma}\label{l:detect-D}
	$d(D)=2$.
\end{lemma}
\begin{proof}
	Suppose we have yet to examine at least $n+1$ edges of the input tournament $G$ and that we have still not found a copy of $D$. 
	Then, there are four vertices $a,b,c,d$ such that none of the edges on pairs $ab$, $bc$, $cd$, have been examined. No matter how the other three (possibly examined) edges on the pairs  $ac$, $ad$, $bd$ are oriented, it is immediate to check that it is possible to orient the edges $ab$, $bc$, $cd$ so that $\{a,b,c,d\}$ induce a copy of $D$.
	Thus, we might need to examine at least $\binom{n}{2}-n$ edges to  determine whether $G$ contain a copy of $D$, implying that $d(D) \ge 2$.
	
	For the upper bound, we will use Lemma \ref{l:D-free}. We first check whether $G$ is transitive, and if not, exhibit a directed triangle on vertices $a,b,c$ with with $(a,b),(b,c),(c,a) \in E(G)$.
	This is straightforward to do in $O(n^2)$ time. Next, for each $S \subseteq \{a,b,c\}$  we construct each of the eight sets $N_S$. Again, this is straightforward to do in $O(n^2)$ time.
	Each of the items 1-6 in the statement of Lemma \ref{l:D-free} can be easily checked to hold in $O(n^2)$ time.
	The only item for which it is not obvious is Item 7. For this we proceed as follows.
	Suppose $X$ and $Y$ are two disjoint sets of vertices, where both $X$ and $Y$ are transitive.
	Suppose also that $X=\{x_1,\ldots,x_p\}$ and $Y=\{y_1,\ldots,y_q\}$ where $(x_i,x_j) \in E(G)$ whenever
	$i < j$ and $(y_i,y_j) \in E(G)$ whenever $i < j$. We must therefore determine whether there exist $x_i,x_j,y_m$ with $i < j$ and with $(x_i,y_m) \in E(G)$ and $(y_m,x_j) \in E(G)$ (symmetrically, we must also determine whether there exist $y_i,y_j,x_m$ with $i < j$ and with $(y_i,x_m) \in E(G)$ and $(x_m,y_j) \in E(G)$, so we only consider the former case). We show how to do this in $O(pq) \le O(n^2)$ time.
	Let $X_i=\{x_1,\ldots,x_i\}$. For each $1 \le m \le q$ let $\alpha_{i,m}$ be the number of in-neighbors of $y_m$ in $X_i$ and let $\beta_{i,m}$ be the number of out-neighbors of $y_m$ in $X \setminus X_i$.
	We must therefore decide if there exists some $(i,m)$ for which both $\alpha_{i,m}$ and $\beta_{i,m}$
	are positive. Computing $\alpha_{1,m}$ and $\beta_{1,m}$ for all $m$ is done in $O(pq)$ time by scanning
	all edge between $X$ and $Y$. Having computed $\alpha_{i-1,m}$ and $\beta_{i-1,m}$, we can compute
	$\alpha_{i,m}$ and $\beta_{i,m}$ in $O(q)$ time as follows. If $(x_i,y_m) \in E(G)$ then
	$\alpha_{i,m} = 1 + \alpha_{i-1,m}$ and $\beta_{i,m} = \beta_{i-1,m}$.
	If $(y_m,x_i) \in E(G)$ then
	$\alpha_{i,m} = \alpha_{i-1,m}$ and $\beta_{i,m} = \beta_{i-1,m}-1$. Thus, we can determine if Item 7
	holds in $O(n^2)$ time. By Lemma \ref{l:D-free}, it follows that we can determine if $G$ has a copy of $D$ in $O(n^2)$ time, showing that $d(D) \le 2$.
\end{proof}

We now proceed to the counting problem. We will use a method similar to the subgraph-equation method of
\cite{KKM-2000}. Recall that $d^+(v)$ and $d^-(v)$ respectively denote the out-degree and in-degree of $v$. For two distinct vertices $u,v$ let $d^+(u,v)$ denote their common out-degree, that is the number of vertices $w$ such that $(u,w) \in E(G)$ and $(v,w) \in E(G)$. Analogously, let $d^-(u,v)$ denote their common
in-degree. Finally, let $p(u,v)$ denote the number of paths of length $2$ from $u$ to $v$, that is 
the number of vertices $w$ such that $(u,w) \in E(G)$ and $(w,v) \in E(G)$. Observe that $p(v,u)$ may differ from $p(u,v)$. Let $A^+$ be the $0/1$ adjacency matrix of $G$ with $A^+(i,j)=1$ if and only if $(i,j) \in E(G)$.
Let $A^-=J-I-A^+$ be the $0/1$ ``incoming'' adjacency matrix of $G$ with $A^-(i,j)=1$ if and only if $(j,i) \in E(G)$ (here $J$ denotes the all-one $n \times n$ matrix and $I$ denotes the $n \times n$ identity matrix).
\begin{lemma}\label{l:four-count}
	$c(T) \le \omega$ for each tournament $T$ on four vertices.
\end{lemma}
\begin{proof}
	Let $u,v$ be a pair of distinct vertices. Consider the $(u,v)$ entry of the following three integer matrix products:
	$(A^+)(A^+)^T$, $(A^-)(A^-)^T$, $(A^+)^2$.
	The first determines $d^+(u,v)$, the second determines $d^-(u,v)$ and the third determines $p(u,v)$.
	Thus, all of these values, for all pairs of vertices, can be computed in $O(n^\omega)$ time.
	
	Consider some copy of $T_4$ in $G$, say on vertices $u,v,w,x$ with $(u,w),(u,x),(v,w),(v,x) \in E(G)$. Observe that $w$ and $x$ are in the common out-neighborhood of $u,v$. Hence,
	$\binom{d^+(u,v)}{2}$ counts this $T_4$ copy precisely once.
	Denoting by $\#T$ the number of copies of $T$ in $G$ it follows that
	\begin{equation}\label{e:t4}
		\#T_4 = \sum_{(u,v) \in E(G)} \binom{d^+(u,v)}{2}\;.
	\end{equation}
	Next, consider four vertices $u,v,w,x$ with $(u,w),(v,w),(x,u),(x,v) \in E(G)$. No matter how the edge
	between $u$ and $v$ is oriented, we have that if $(x,w) \in E(G)$ then these four vertices induce a $T_4$ and if $(w,x) \in E(G)$ then they induce an $X$. It is also easy to verify that every copy of $T_4$ 
	has precisely one such pair $u,v$ and every copy of $X$ has one such pair $u,v$.
	We therefore have that 
	\begin{equation}\label{e:t4+x}
		\#T_4 +\#X= \sum_{(u,v) \in E(G)} {d^+(u,v)}\cdot d^-(u,v)\;.
	\end{equation}
	Next, consider four vertices $u,v,w,x$ with $(u,v),(u,w),(v,w),(u,x),(x,v) \in E(G)$. If $(x,w) \in E(G)$ then these four vertices induce a $T_4$ and if $(w,x) \in E(G)$ then they induce a $D$. It is also easy to verify that every copy of $T_4$  has precisely one such ordered pair $u,v$ and every copy of $D$ has one such ordered pair $u,v$ (in both $T_4$ and $D$, $u$ plays the role of the source vertex).
	We therefore have that 
	\begin{equation}\label{e:t4+d}
		\#T_4 +\#D= \sum_{(u,v) \in E(G)} {d^+(u,v)}\cdot p(u,v)\;.
	\end{equation}
	Finally, we have the obvious equation	
	\begin{equation}\label{e:all-four}
		\#T_4 +\#X + \#D +\#D^T= \binom{n}{4}\;.
	\end{equation}
	Equations \eqref{e:t4}, \eqref{e:t4+x}, \eqref{e:t4+d}, \eqref{e:all-four} form a system of linear equations in the variables $\#T_4, \#X, \#D, \#D^T$ whose coefficient matrix is
	\[ \begin{pmatrix}	1 & 0 & 0 & 0 \\ 1 & 1 & 0 & 0 \\ 1 & 0 & 1 & 0 \\ 1 & 1 & 1 & 1 \end{pmatrix} \]
	which is nonsingular. Thus, $\#T$ for each tournament $T$ on four vertices can be computed in $O(n^\omega)$ time. It follows that $c(T) \le \omega$.
\end{proof}

\begin{proof}[Proof of Theorem \ref{t:four}]
	By Lemma \ref{l:detect-c4} and Lemma \ref{l:detect-D} we have that $d(T)=2$ for each tournament on four vertices other than $T_4$ (for which $d(T_4)=0$). Thus, $d(4)=2$.
	By Lemma \ref{l:four-count}, $c(4) \le \omega$.
\end{proof}

\section{Counting five-vertex tournaments}\label{sec:five}

Figure \ref{f:tour-5} lists and names the twelve tournaments on five vertices. Observe that $H_1^T$ and $H_2^T$ are the transpose of $H_1$ and $H_2$, respectively. The other eight tournaments are isomorphic to their transpose. The unique regular tournament on five vertices is denoted by $R_5$. Also notice
that $H_4$ and $H_5$ (while non-isomorphic) have the same out-degree sequence and $H_6,H_7,H_8$ also
have the same out-degree sequence.

\begin{figure}[t]
	\includegraphics[scale=0.6,trim=34 200 150 38, clip]{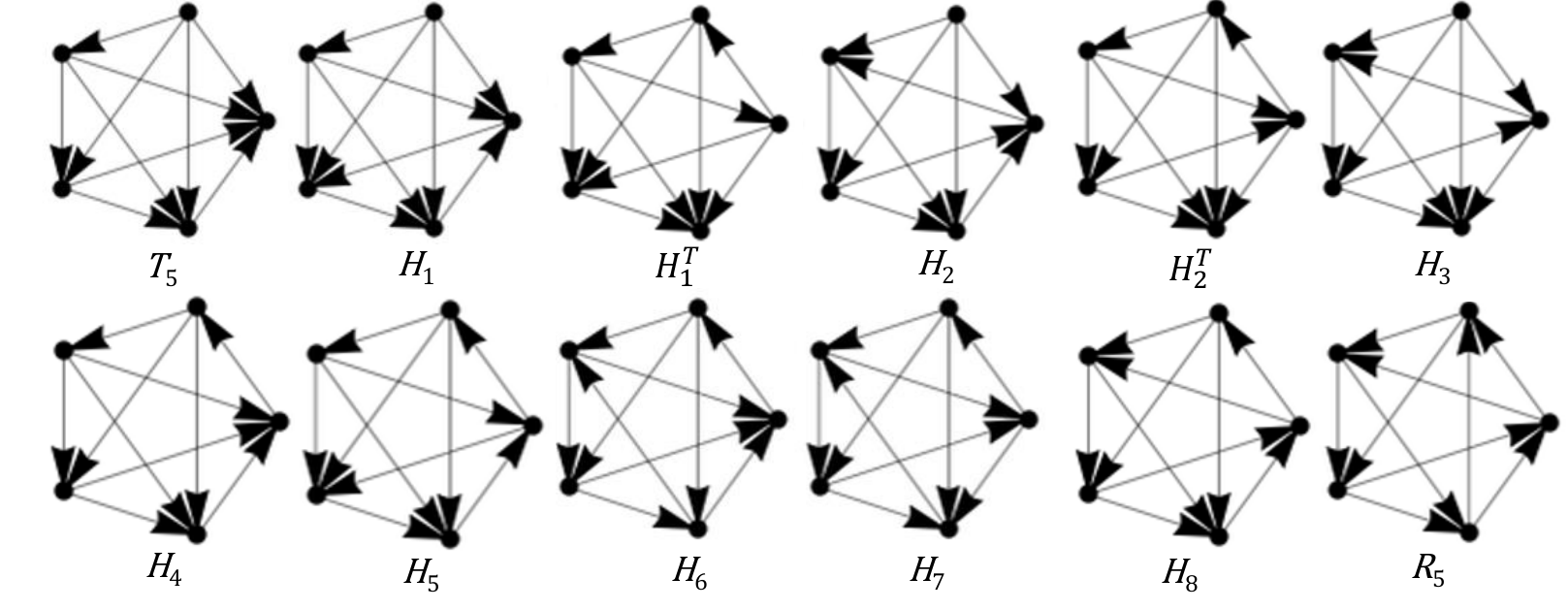}
	\caption{The tournaments on $5$ vertices.}
	\label{f:tour-5}
\end{figure} 

As in the previous section, for two distinct vertices $u,v$ in a tournament, let $d^+(u,v)$ denote their common out-degree, let $d^-(u,v)$ denote their common in-degree and let $p(u,v)$ denote the number of paths of length $2$ from $u$ to $v$. We can obtain a system of $20=\binom{4}{1}+\binom{4}{1}\cdot\binom{3}{1} + \binom{4}{3}$ linear equations whose right hand side equals a polynomial of degree $3$ in $p(u,v),p(v,u),d^+(u,v)$ and $d^-(u,v)$, see Table \ref{table:1}.
The first column in Table \ref{table:1} is the equation number. The twelve intermediate columns correspond
to variables, one for each count of a tournament on five vertices. The entries are coefficients of these variables in the corresponding linear equation (so the inner $20 \times 12$ cells correspond to the coefficient matrix of the linear system; to avoid clutter, zero coefficients are not listed). The right column in Table \ref{table:1} corresponds to the right hand side of the linear equation. For an entry $P$ in that column, it reads as $\sum_{(u,v) \in E(G)} P$.

\begin{table}[ht!]
	\centering
	\renewcommand{\arraystretch}{1.4}
	\begin{tabular}{c||cccccccccccc||c}
		eq  & \hspace{-8pt}  $\#T_5$ \hspace{-12pt} & $\#H_1$  \hspace{-12pt} & $\#H_1^T$  \hspace{-12pt} & $\#H_2$  \hspace{-12pt} & $\#H_2^T$  \hspace{-12pt} & $\#H_3$  \hspace{-12pt} & $\#H_4$  \hspace{-12pt} & $\#H_5$  \hspace{-12pt} & $\#H_6$  \hspace{-12pt} & $\#H_7$  \hspace{-12pt} & $\#H_8$  \hspace{-12pt} & $\#R_5  \hspace{-4pt}$ & $=\sum_{(u,v) \in E(G)}$ of \\
		\hline
		$1$ & $1$ & $1$ & & & & & & & & & & & $\binom{d^+(u,v)}{3}$ \\
		$2$ & $1$ & & $1$ & & & & & & & & & & $\binom{d^-(u,v)}{3}$ \\
		$3$ & & & & & & & $1$ & & $1$ & & & & $\binom{p(v,u)}{3}$ \\
		$4$ & $1$ & & & & & $1$ & & & & & & & $\binom{p(u,v)}{3}$ \\
		$5$ & $1$ & & & & $1$ & & $1$ & & & & $1$ & & $\binom{d^+(u,v)}{2}d^-(u,v)$ \\
		$6$ & $1$ & & & $1$ & & & $1$ & & & & $1$ & & $\binom{d^-(u,v)}{2}d^+(u,v)$ \\
		$7$ & $1$ & & & $2$ & & $3$ & & & & & & & $\binom{d^+(u,v)}{2}p(u,v)$ \\
		$8$ & & & $3$ & & $1$ & & $1$ & $1$ & & & & & $\binom{d^+(u,v)}{2}p(v,u)$ \\
		$9$ & $1$ & & & & $2$ & $3$ & & & & & & & $\binom{d^-(u,v)}{2}p(u,v)$ \\
		$10$ & & $3$ & & $1$ & & & $1$ & $1$ & & & & & $\binom{d^-(u,v)}{2}p(v,u)$ \\
		$11$ & & & & & & & $2$ & $1$ & & & $1$ & & $\binom{p(u,v)}{2}p(v,u)$ \\
		$12$ & $1$ & & $3$ & & $2$ & & & & & & & & $\binom{p(u,v)}{2}d^-(u,v)$ \\
		$13$ & $1$ & $3$ & & $2$ & & & & & & & & & $\binom{p(u,v)}{2}d^+(u,v)$ \\
		$14$ & & & & & & & & & & $1$ & $2$ & $5$ & $\binom{p(v,u)}{2}p(u,v)$ \\
		$15$ & & & & $1$ & & & & $1$ & & $1$ & $1$ & & $\binom{p(v,u)}{2}d^-(u,v)$ \\
		$16$ & & & & & $1$ & & & $1$ & & $1$ & $1$ & & $\binom{p(v,u)}{2}d^+(u,v)$\\
		$17$ & $1$ & $3$ & $3 $ & & & & $1$ & $3$ & & $1$ & & & $d^+(u,v)d^-(u,v)p(u,v)$ \\
		$18$ & & & & $1$ & $1$ & $3$ & & & $3$ & $2$ & $1$ & $5$ & $d^+(u,v)d^-(u,v)p(v,u)$ \\
		$19$ & & & & & $2$ & & $1$ & $1$ & $3$ & $2$ & $1$ & & $d^+(u,v)p(u,v)p(v,u)$\\
		$20$ & & & & $2$ & & & $1$ & $1$ & $3$ & $2$ & $1$ & & $d^-(u,v)p(u,v)p(v,u)$
	\end{tabular} 
\caption{Linear equations involving counts of copies of the twelve tournaments on five vertices.}\label{table:1}
\end{table} 

While a bit lengthy, the correctness of each equation is easy to verify. Let us consider a representative example from each r.h.s. ``type'', e.g. Equations $1,7,18$ of Table \ref{table:1}.
For Equation $1$, we are counting those five-vertex tournaments in $G$ having two vertices $u,v$ that dominate
(i.e. send edges to) three other vertices. Now, these three other vertices may induce a $T_3$,
in which case the tournament counted is $T_5$, or may induce a $C_3$, in which case the tournament counted is $H_1$. Hence, $\#T_5+\#H_1 = \sum_{(u,v) \in E(G)}\binom{d^+(u,v)}{3}$.
In Equation $7$ we are counting five-vertex tournaments in $G$ that have two vertices $u,v$ with $(u,v)$ an edge, two additional vertices dominated by both $u$ and $v$, and an additional vertex
dominated by $u$ and that dominates $v$. Namely, tournaments that contain as a spanning subgraph the directed
graph $X$ depicted in the left side of Figure \ref{f:subgraphs}. Notice that $X$ is a spanning subgraph appearing once in $T_5$, twice in $H_2$, and three times in $H_3$. Hence,
$\#T_5+2\#H_2+3\#H_3=\sum_{(u,v) \in E(G)}\binom{d^+(u,v)}{2}p(u,v)$.
For Equation $18$, we are counting five-vertex tournaments in $G$ that have two vertices $u,v$ with $(u,v)$ an edge, one additional vertex dominated by both $u$ and $v$, one additional vertex dominating both $u$ and $v$ and one additional vertex
dominated by $v$ and that dominates $u$. Namely, tournaments that contain as a spanning subgraph the directed
graph $Y$ depicted in the right side of Figure \ref{f:subgraphs}. Notice that $Y$ is a spanning subgraph appearing once in each of $H_2$, $H_2^T$, $H_8$, twice in $H_7$, three time in each of $H_3,H_6$ and five times in $R_5$.  
Hence, $\#H_2+\#H_2^T+\#H_8+2\#H_7+3\#H_3+3\#H_6+5\#R_5=\sum_{(u,v) \in E(G)}d^+(u,v)d^-(u,v)p(v,u)$.

\begin{figure}[t]
	\includegraphics[scale=0.8,trim=-20 380 250 117, clip]{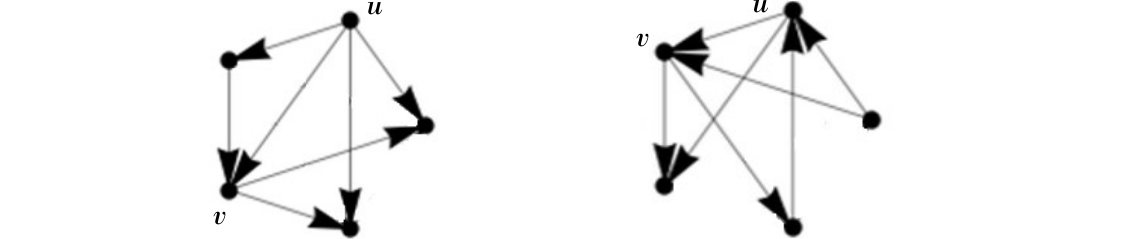}
	\caption{The directed graphs corresponding to the r.h.s. of Equation 7 (left) and 18 (right).}
	\label{f:subgraphs}
\end{figure} 

As shown in the first paragraph of the proof of Lemma \ref{l:four-count}, the values in the r.h.s. of Table \ref{table:1} can all be computed in $O(n^\omega)$ time. Unfortunately, however, the $20 \times 12$ coefficient matrix, denoted hereafter by $A$, only has rank $10$, so we cannot guarantee a unique solution.
Our approach would be to show that some submatrix of $A$ has full column rank, while for columns not in that sub-matrix, we shall count the corresponding tournaments differently.
We shall first need the following lemma, which is applied to the case $k=5$.
\begin{lemma}\label{l:count-extension}
	Let $T$ be a tournament on $k$ vertices having a dominating vertex (a source) or having a dominated vertex (a sink). Then $c(T) \le 1 + c(k-1)$.
\end{lemma}
\begin{proof}
	Assume that $T$ has a dominating vertex, denoted by $a$ (the case of $T$ having a sink vertex is symmetrical). Let $T^*$ be the sub-tournament obtained from $T$ after removing $a$.
	For a given tournament $G$, and a given vertex $v \in V(G)$, let $G_v$ denote the sub-tournament
	of $G$ induced by the out-neighbors of $v$. Observe that $T$-COUNT for the instance $G$ can be solved
	by summing for all $v \in V(G)$ the results of ${T^*}$-COUNT for the corresponding instance $G_v$.
	Indeed, each copy of $T$ in $G$ maps $a$ to some vertex $v \in V(G)$ and maps the remaining vertices to a unique copy of $T^*$ in $G_v$. As $T^*$ is a tournament on $k-1$ vertices, it follows that
	$c(T) \le 1+c(k-1)$.
\end{proof}
\begin{corollary}\label{coro:1}
	For each $T \in \{T_5, H_1, H_1^T, H_2, H_2^T, H_3 \}$ it holds that $c(T) \le \omega+1$.
\end{corollary}
\begin{proof}
	Observe that each of the listed tournaments has a dominating vertex or a dominated vertex.
	The corollary now follows from Lemma \ref{l:count-extension} by recalling that $c(4) \le \omega$,
	as shown in Theorem \ref{t:four}.
\end{proof}
Consider the sub-matrix of $A$ consisting only of the columns {\em not} corresponding to those listed
in Corollary \ref{coro:1}. This is a $20 \times 6$ matrix but it is still not of full column rank; its rank is only $5$. Fortunately, there is another tournament, in fact it is $H_8$, for which we can compute
$H_8$-COUNT differently.
\begin{lemma}\label{l:H_8}
	$c(H_8) \le \omega+1$.
\end{lemma}
\begin{proof}
	We label the vertices of $H_8$ as follows. Let $a$ be the vertex with out-degree $3$. Let $d$ be the
	vertex with out-degree $1$. Let $e$ be the unique in-neighbor of $a$.
	Let $b$ be the other out-neighbor of $e$ and let $c$ be the remaining vertex. Observe that $e$ has precisely two out-neighbors ($a$ and $b$) and precisely two in-neighbors
	($c$ and $d$) and that each out-neighbor of $e$ dominates each in-neighbor of $e$.
	It is important to observe that this situation (i.e. of having two out-neighbors and two in-neighbors and both out-neighbors dominate both in-neighbors) is {\em unique} for $H_8$ and its vertex $e$.
	It does not hold in any other tournament on five vertices.
	Now, suppose $G$ is a given tournament for which we should solve $H_8$-COUNT.
	For each vertex $v \in V(G)$, let $A_v$ be the the set of out-neighbors of $v$ and let
	$B_v$ be the set of in-neighbors of $v$. Construct an undirected bipartite graph with sides $A_v,B_v$
	keeping (as undirected edges) only the edges of $G$ going from $A_v$ to $B_v$.
	Observe that every undirected four-cycle in $G$, say on vertices $u_1,u_2 \in A_v$ and $w_1,w_2 \in B_v$
	uniquely corresponds to a copy of $H_8$ in $G$. Indeed, this can be seen by the mapping
	sending $e$ to $v$, sending $a$ and $b$ to $u_1$ and $u_2$ (it does not matter which is mapped to which as
	reversing the edge $(a,b)$ in $H_8$ is an automorphism), and sending $c$ and $d$ to
	$w_1$ and $w_2$ (it does not matter which is mapped to which as
	reversing the edge $(c,d)$ in $H_8$ is an automorphism).
	Now, the number of undirected four-cycles in a graph with $O(n)$ vertices can be computed in $O(n^\omega)$ time, see \cite{AYZ-1997}. Performing this count for each $v \in V(G)$ and summing the outcomes gives the number of $H_8$ in $G$ in $O(n^{\omega+1})$ time.
\end{proof}

By Lemma \ref{l:H_8} and Corollary \ref{coro:1}, it only remains to show how to solve $T$-COUNT for
$T \in \{H_4, H_5, H_6, H_7, R_5 \}$, given the counts of $T \in \{T_5, H_1, H_1^T, H_2, H_2^T, H_3, H_8\}$
and given the values in the r.h.s. of Table \ref{table:1}. This can be done in constant time since the sub-matrix of $A$ on the columns $\{H_4, H_5, H_6, H_7, R_5 \}$ has full column rank. Indeed, to see this just observe that its square sub-matrix on rows $3,5,8,15,18$ being
	\[ \begin{pmatrix}	1 & 0 & 1 & 0 & 0 \\ 1 & 0 & 0 & 0 & 0 \\ 1 & 1 & 0 & 0 & 0 \\ 0 & 1 & 0 & 1 & 0 \\
	0 & 0 & 3 & 2 & 5 \end{pmatrix} \]
is non-singular. \qed

\section*{Acknowledgment}
The author thanks Lior Gishboliner for useful comments.

\end{document}